\documentclass[12pt]{article}
\usepackage{amssymb,amsfonts,amsmath,amsthm}
\usepackage{simplemargins}
\usepackage{ulem}
\newtheorem{theorem}{Theorem}[section]
\newtheorem{lemma}[theorem]{Lemma}

\newtheorem{proposition}[theorem]{Proposition}
\newtheorem{corollary}[theorem]{Corollary}

\theoremstyle{definition}
\newtheorem{definition}[theorem]{Definition}
\newtheorem{example}[theorem]{Example}

\theoremstyle{remark}
\newtheorem{remark}[theorem]{Remark}

\usepackage{amssymb,amsfonts,amsmath,amsthm}

\def\diag{\operatorname{diag}}
\def\rank{\operatorname{rank}}
\def\RR{\mathbb{R}}
\def\ZZ{\mathbb{Z}}

\def\min{\mathrm{min}}
\def\bsigma{{\Bar{\sigma}}}
\def\mtx#1#2#3#4{\left[\begin{array}{cc}#1 & #2\\ #3 & #4 \end{array}\right]}

\def\Tr{\operatorname{tr}}

\def\Ker{\operatorname{Ker}}
\def\Im{\operatorname{Im}}

\newtheorem*{main-question}{{\bf Main Question}}

\def\CC{\mathbb{C}}

\def\OOplus={{ {{{\mathcal{O}}}}_+}}
\def\OObar={\overline {\mathcal{O}_+}}
\def\RR{\mathbb{R}}

\def\IR{\mathbb{R}}
\def\IZ{\mathbb{Z}}
\def\Z{\mathbb{Z}}

\def\E{\mathcal{E}}
\def\sE{\mathcal{E}}

\def\sN{\mathcal{N}}

\def\e{\varepsilon}

\def\one{\mathbf{1}}

\newcommand{\od}{\stackrel{\mbox {\tiny {def}}}{=}}
\def\tr{\operatorname{tr}}

\newcommand{\supp}{\operatorname{supp}}

\def\bsigma{{\Bar{\sigma}}}
\def\mtx#1#2#3#4{\left[\begin{array}{cc}#1 & #2\\ #3 & #4 \end{array}\right]}

\def\vec#1#2{\left(\begin{array}{c}#1\\ #2 \end{array}\right)}

\usepackage{multicol}
\usepackage{color}
\definecolor{gold}{rgb}{0.85,.66,0}
\definecolor{cherry}{rgb}{0.9,.1,.2}
\definecolor{burgundy}{rgb}{0.8,.2,.2}
\definecolor{orangered}{rgb}{0.85,.3,0}
\definecolor{orange}{rgb}{0.85,.4,0}
\definecolor{olive}{rgb}{.45,.4,0}
\definecolor{lime}{rgb}{.6,.9,0}
\definecolor{green}{rgb}{.2,.7,0}
\definecolor{grey}{rgb}{.4,.4,.2}
\definecolor{brown}{rgb}{.4,.2,.1}
\definecolor{blue}{rgb}{0,.0, .81}
\definecolor{red}{rgb}{1,0,0}

\usepackage{ulem}
\def\thmI{Let $(W,D)$ be a threshold-linear network.  A subset of neurons
  $\sigma$ is a stable set of $(W,D)$ if and only if the principal
  submatrix $(-D+W)_\sigma$ is stable.  Similarly, $\sigma$ is a marginal set or an unstable set
  of $(W,D)$ if and only if $(-D+W)_\sigma$ is marginally stable or unstable, respectively.}

\def\thmII{A threshold-linear network is maximally flexible in $\sN(n)$ if and only if it is a rank $1$ network.}

\def\thmIII{Let $(J,D)_G$ be a maximally flexible threshold-linear
  network in $\sN(G)$, and suppose that the clique complex $X(G)$ 
  satisfies $H_1(X(G);\ZZ) = 0$.  Then $(J,D)_G$ has a rank 1 completion.  In particular, $(J,D)_G$ has no silent connections.}

\def\thmV{All rank 1 threshold-linear networks on $n$ neurons are maximally flexible in $\sN(n)$, and have flexibility $2^n-n-1$.  All $G$-constrained threshold-linear networks with a rank 1 completion are maximally flexible in $\sN(G)$, and have flexibility $|X(G)|-n-1$.}

\title{Flexible memory networks\footnote{Accepted to Bulletin of Mathematical Biology, 11 July 2011.}}
\author{Carina Curto$^1$, Anda Degeratu$^2$, Vladimir Itskov$^1$}
\date{}

\begin{document}

\maketitle

\vspace{-.3 in}
\noindent \hspace{.75in} $^1$ Department of Mathematics, University of Nebraska-Lincoln

\noindent \hspace{.75in} $^2$ Max Planck Institute for Gravitational Physics, Golm, Germany


\begin{abstract}
Networks of neurons in some brain areas are flexible enough to encode new memories quickly.  Using a standard firing rate model of recurrent networks, we develop a theory of flexible memory networks.  Our main results characterize networks having the maximal number of flexible memory patterns, given a constraint graph on the network's connectivity matrix.  Modulo a mild topological condition, we find a close connection between maximally flexible networks and rank $1$ matrices.  The topological condition is $H_1(X;\ZZ)=0$, where $X$ is the clique complex associated to the network's constraint graph; this condition is generically satisfied for large random networks that are not overly sparse.  In order to prove our main results, we develop some matrix-theoretic tools and present them in a self-contained section independent of the neuroscience context.
\end{abstract}

\begin{small}
\tableofcontents
\end{small}

\section{Introduction}
New memories in some brain areas can be encoded quickly \cite{Rutishauser2006}. 
It is widely believed that memories are stored via
changes in the synaptic efficacies between neurons. Irrespective of the plasticity mechanism, or 
`learning rule', used to encode memory patterns, rapid learning is perhaps most easily
accomplished if new patterns can be learned via only small changes in connection strengths between neurons. It
may thus be desirable for fast-learning, 
flexible networks to have architectures that enable many
memory patterns to be encoded (and unencoded) by only small perturbations of the synaptic
connections.  Here by `architecture' we mean the pattern of synaptic strengths, or {\em weights}, assigned to directed connections between neurons.  Which network architectures allow maximal flexibility for learning and unlearning new memories?

We study this question in the context of a standard firing rate model of recurrent neural networks.    Building on the framework of  `permitted' and `forbidden' sets first introduced in \cite{Seung:2003}, we think of the recurrent network as a gating device that allows only a restricted set of patterns, the stored `memories', to be activated by external feed-forward input.  
In Theorem~\ref{thm:stable-clique}, we establish a correspondence between the memory patterns encoded by a recurrent network and the set of stable principal submatrices of the network's effective connectivity matrix.  We then make precise the notion of memory patterns that are `flexible' in the sense that they can be encoded (learned) and unencoded (forgotten) via only small changes to the network weights.  Our main results, Theorems~\ref{thm:rank1}, \ref{thm:main-result} and \ref{thm:unconstrained}, characterize network architectures with the maximal number of flexible memories.  

\subsection{Network dynamics and architecture}
We consider a standard firing rate model \cite{DayanAbbott, ErmentroutTerman} with heterogeneous timescales,
\begin{equation*} \frac{d x_i}{dt}=-\frac1{\tau_i}x_i+\varphi\left(\sum_{j=1}^n  W_{ij}x_j+b_i  \right)\!\!, \;\;\; \mathrm{for}\; \;i=1,...,n,
\end{equation*}
where $n$ is the number of neurons.  The real-valued function $x_i = x_i(t)$ is the firing rate of the $i$th neuron, $b_i$ is the external input to the $i$th neuron, and $W_{ij}$ denotes the effective strength of the recurrent connection from the $j$th to the $i$th neuron.  The timescale $\tau_i$ gives the rate of recovery to rest in the absence of external or recurrent inputs.  The nonlinear function $\varphi:\RR \rightarrow \RR_{\geq 0}$ satisfies $\varphi(y) = 0$ whenever $y \leq 0$, and ensures that the firing rates $x_1,...,x_n$ are non-negative.  Although the threshold appears to be zero for all neurons, heterogeneous thresholds can easily be incorporated into the $b_{i}$s.
Note that  $\tau_i > 0$, while $b_i$ and $W_{ij}$ can take on both positive and negative values.  The dynamics of the network can be described more compactly as
\begin{equation}\label{eq:dynamics}
\dot x = -Dx+\varphi\left(W x + b  \right)\!,
\end{equation}
where $D\od \operatorname{diag}({\tau_1}^{-1},...,{\tau_n}^{-1})$ is the diagonal matrix of inverse time constants, and $\varphi$ is applied to each coordinate when the argument is a vector.  
Note that we do not require that the matrix $W$ respect Dale's law\footnote{Dale's law states that every element in the same column of the connectivity matrix must have the same sign  \cite{DayanAbbott}.  This is because neurons have either purely excitatory or purely inhibitory synapses onto other neurons.}, as the entries are considered to be {\em effective} connection strengths between {\em principal} (excitatory) neurons.  Negative weights thus reflect effectively inhibitory interactions, mediated by the presence of non-specific interneurons that do not otherwise enter into the model.  We always assume that the diagonal entries of the matrix $-D + W$ are strictly negative; otherwise, individual neurons may experience a run-away excitatory drive even in the absence of external or recurrent inputs.  We will also assume that $b_i$ is constant in time, though it may vary across neurons.  For a given choice of nonlinearity $\varphi$, the network and its dynamics~\eqref{eq:dynamics} are denoted by the pair of matrices $(W,D)$.


It is worth noting here that we regard the model  \eqref{eq:dynamics} as a description of {\it fast-timescale dynamics}.  A more realistic network would also include stochastic fluctuations and adaptive mechanisms on a slower timescale \cite{Abbott:2004:Nature}, so that `fixed points' of the fast-timescale dynamics only appear for short periods of time.  Such fixed points will serve as our model for (transiently) activated memory patterns.  

To study network flexibility, we will think of the matrix of effective connection strengths $W$ as a sum of two components,
$$W = J + A,$$
where $J$ corresponds to a fixed and underlying {\em architecture}, and $A$ is a matrix of {\em perturbations} about $J$. 
While $J$ reflects broad patterns of connection strengths that may be conserved across animals or across time, the matrix $A$ captures individual variations, and is constantly changing as a function of the animal's learning and experience.  Our main question is then: 

\begin{main-question}[version 1]
  What architectures $J$ allow maximal flexibility for learning and
  unlearning new memory patterns under small perturbations $A$?
\end{main-question}

\noindent We will consider this as a question about the strengths, or {\em weights}, of the recurrent connections between neurons, rather than as a question about which neurons are connected.  The pattern of allowed connections between neurons will be treated as a constraint.  Indeed, networks with different {\em strengths} of connections, but identical connectivity patterns, may have significantly different dynamics and attractors.  Moreover, in biological neural circuits the anatomical connectivity may be difficult to modify, but the weights of synapses are known to change on relatively short timescales in response to learning and experience.

\subsection{Memory patterns as `stable sets' in threshold-linear networks}

Before addressing our main question about perturbations of network architectures, we investigate the set of
memory patterns corresponding to any fixed network $(W,D)$.  The idea of stable fixed points as a model for stored memory patterns in recurrent networks dates back at least to \cite{Hopfield:1982}.  Following the framework of  \cite{Seung:2003}, subsets of neurons that are active at stable fixed points of~\eqref{eq:dynamics} will serve as our model for stored memory patterns.

Recall that a fixed point $x^*$ is {\em asymptotically stable} if there exists an open neighborhood $U$ of $x^*$ such that $\lim_{t\to\infty}x(t)=x^*$ for every trajectory $x(t)$ with the initial condition $x(0)\in U.$ 
If the fixed point $x^*$ has only the property that for all nearby initial
  conditions $x(0)\in U$ the trajectory $x(t)$ remains very close to $x^*$ for
  all later times, then $x^*$ is a {\em stable} point of the
  network dynamics. Note that every asymptotically stable point is
  stable, but the converse is not true.
For a given firing rate vector $x \in \RR^n_{\geq 0}$, we call the subset of active neurons the {\em support} of $x$,
$$\supp(x)\od \{i \;|\; x_{i} > 0\} \subset \{1,...,n\}.$$

\begin{definition}[stable, marginal, and unstable sets]\label{def:stable-set}
Let $(W,D)$ be a network on $n$ neurons with nonlinearity $\varphi$.  
A non-empty subset of neurons $\sigma \subset \{1,...,n\}$ is a {\em stable set} of $(W,D)$ if there exists an asymptotically stable fixed point $x^{*}$ of the dynamics \eqref{eq:dynamics} such that $\supp(x^{*}) = \sigma$, for at least one external input vector $b \in \RR^n$.   
A {\em marginal set} of $(W,D)$ is a non-empty subset of neurons
$\sigma$ for which there exists a
stable fixed point of the dynamics (but no asymptotically stable fixed point) with support $\sigma$ for at
least one external input vector $b$, and
an {\em unstable set} of $(W,D)$ is a non-empty subset of neurons that is neither stable
nor marginal.\footnote{Stable and unstable sets were previously introduced in  \cite{Seung:2003}, where
they were called ``permitted'' and ``forbidden'' sets, respectively.}
\end{definition}

\noindent Stable sets are our model for memory patterns encoded by the
network.  For a fixed external input $b$ there may be one, many, or no
asymptotically stable fixed points.  As we range over all possible
inputs, however, we obtain the set of stable sets of the network.
Clearly, there can be at most $2^n-1$ stable sets in a network of $n$
neurons.  In cases of interest, however, the recurrent network
performs meaningful computations precisely because only a small
fraction of subsets are stable \cite{Seung:2003}.  Note that a pair of neurons in a stable set need not be connected (i.e., we may have $W_{ij} = 0$ for $i,j$ in a stable set).  

In general, it is very difficult to determine analytically the stable fixed points for a
high-dimensional, nonlinear dynamical system.  If the nonlinearity
$\varphi$ in \eqref{eq:dynamics} is threshold-linear, however, it is possible
to use standard tools from linear systems of ordinary
differential equations to obtain exact results.  For this reason,
we now restrict ourselves to {\em threshold-linear networks}, which
are networks $(W,D)$ where the nonlinearity is chosen as
\begin{eqnarray*}
 \varphi(y) = [y]_+ \od 
  \left\{ \begin{array}{ccc} y & \mbox{if} & y >0, \\ 
      0 & \mbox{if} & y \leq 0. \end{array}\right.
\end{eqnarray*}
Although sigmoids more closely match experimentally measured
input-output curves for neurons, the above
threshold-nonlinearity is often a good approximation when neurons are
far from saturation \cite{DayanAbbott, Geffen2009}.  If we assume that encoded memory patterns are realized by
 neurons firing sufficiently below saturation, it is reasonable to
model them as stable sets of the threshold-linear
dynamics:
\begin{equation}\label{eq:threshlin}
 \dot x = -Dx+\left[W x + b \right]_+.
\end{equation}

In a (nondegenerate) linear system, $\dot x = (-D+W)x + b$, there can be at most one fixed point of the dynamics for a given input vector $b \in \RR^n$; its stability is characterized by the eigenvalues of the matrix $-D+W$.  Unlike linear systems, the threshold-linear network \eqref{eq:threshlin} can exhibit multiple fixed points for the same input vector $b$.  It turns out, however, that stable, unstable and marginal sets of neurons in threshold-linear networks have simple characterizations in terms of the eigenvalues of the corresponding {\em principal submatrices} of $-D+W$.  

Given an $n \times n$ matrix $A$, and a subset $\sigma \subset
  \{1,...,n\}$, the {\em principal submatrix} $A_\sigma$ is the square
  matrix obtained by restricting $A$ to the index set $\sigma$;
  i.e., if $\sigma = \{s_1,...,s_k\}$, then $A_\sigma$ is the $k
  \times k$ matrix with $(A_\sigma)_{ij} = A_{s_i s_j}$. 
We call a square matrix {\em stable} if all its eigenvalues have
strictly negative real part.  We call a matrix {\em unstable} if at least one eigenvalue has
strictly positive real part, and {\em marginally stable} if no eigenvalue has
strictly positive real part and at least one eigenvalue is purely imaginary.
Marginally stable matrices are thus on the boundary between stable and unstable
matrices.   

We now state our characterization of stable sets in terms of the stability of principal submatrices:

\begin{theorem}\label{thm:stable-clique}
\thmI
\end{theorem}

\noindent In the case of symmetric threshold-linear networks, where the matrix $W$ is symmetric, the equivalence between stable (`permitted') sets and stable principal submatrices was shown in  \cite{Seung:2003}.  We give the proof of Theorem~\ref{thm:stable-clique} in Section~\ref{sec:stable-set}.

\subsection{G-constrained networks}
There are two ways in which a zero-weight connection between two neurons may arise.
On the one hand, there may be a lack of anatomical connectivity between the neurons.
On the other hand, many synaptic connections that appear anatomically are not functional
- these are referred to as {\em silent synapses} \cite{Kerchner:2008}.  While the first type of zero-weight connection
cannot be perturbed without major changes to the network architecture, silent synapses may become active via small
modifications.  In addressing our main question, we are therefore interested in characterizing maximally flexible networks
where some connections are constrained to be zero, while the remaining weights (some of which may also be zero) can 
be modified by small perturbations of the network. The following definitions hold for general networks, not just threshold-linear ones.

Let $G = (V,E)$ be a simple graph with vertices $V = \{1,...,n\}$ and edges $E$.  
We say that an $n \times n$ architecture matrix $J$ is {\em constrained by} the graph $G$
if $J_{ij} = 0$ for all edges $(ij) \notin E$.  By abuse of notation, we often use $G$ to refer to the edge set $E$.  Note that all
architectures on $n$ neurons are constrained by the complete graph
$G=K_n$.  
If for $(ij) \in G$ the entry $J_{ij} = 0$, we say that
there is a {\em silent connection} from neuron $j$ to neuron $i$.  This mirrors the phenomenon of silent synapses in the brain.

We define an {\em $\varepsilon$-perturbation} of a network architecture $J$ to be a matrix $A$ whose entries all satisfy $|A_{ij}| \leq \varepsilon$. We say that an $\varepsilon$-perturbation is {\em consistent with} $G$ if the matrix $A$ satisfies $A_{ij} = 0$ for all $(ij) \notin G$.  In other words, consistent $\varepsilon$-perturbations can only perturb entries that are not constrained to be zero (including silent connections). 


When considering an architecture $J$ that
is constrained by a graph $G$, we refer to the network
as $(J,D)_G$.  For a given $\varphi$, we
use the following notation for the set of all
$G$-constrained network architectures:
\begin{eqnarray*}
  \sN(G): = \{(J,D)_G\} = \{(J,D)  \; | \; J_{ij} = 0 \mbox{ for all } (ij) \notin G\}.
\end{eqnarray*}
Note that the set of constrained architectures is independent of the nonlinearity $\varphi$. 
When $G = K_n$ is the complete graph (no constraint), we will simply write $\sN(n): = \sN(K_n)$.  If $G_1 \subset G_2$, then $\sN(G_1) \subset \sN(G_2)$.  An $\varepsilon$-perturbation of a network $(J,D)_G$ will always be assumed to be consistent with $G$, and hence to stay within $\sN(G)$.

We can now state our main question a bit more precisely:

\begin{main-question}[version 2]
  For a given constraint graph $G$, what network architectures $(J,D)_G \in \sN(G)$ allow 
  the maximal number of subsets of neurons that can become both stable sets (learned/encoded) and unstable sets (forgotten/unencoded) via arbitrarily small $\varepsilon$-perturbations of $J$?
\end{main-question}

\noindent We call such subsets of neurons {\em flexible} memory patterns.

\subsection{Flexible memory patterns as `flexible cliques'}

Intuitively, a flexible memory pattern is a subset of neurons that can become both a stable set and an unstable set via only small
modifications of the network's connection strengths.  Ideally, these modifications should be specific enough not to change the stability of any other subsets.  
Moreover, we would like flexible memory patterns to correspond to subsets of neurons that are unconstrained in their connections to each other.  In other words, these subsets of neurons should be all-to-all connected in the sense
that all mutual connections can be perturbed, although some may be zero-weight (silent) connections.  
We model such memory patterns as `flexible cliques'; a precise definition is given below.

Recall that a {\em clique} in a graph $G$ is a subset of vertices that are all-to-all connected, and the {\em clique complex} of $G$, denoted $X(G)$, is the set of all cliques.  We will say that  $\sigma \subset \{1,...,n\}$ is a {\em stable clique} of the network
$(W,D)_G$ if $\sigma$ is a stable set  and
$\sigma \in X(G)$.  Similarly, an {\em unstable clique} is an unstable
set $\sigma$ such that $\sigma \in X(G)$, and a {\em marginal clique} is
a marginal set $\sigma$ such that $\sigma \in X(G)$.  
Because the stability of a matrix forces one or more of its principal submatrices
to be stable (see Lemma~\ref{lemma:2x2}), one cannot require that a perturbation that changes the stability of a marginal
clique in a threshold-linear network also preserves all other marginal cliques.
For this reason we introduce the notions of `maximally stable' and `minimally unstable' cliques.
A {\em maximally stable clique} is a stable clique that is not properly contained in
any larger stable clique; a {\em minimally unstable clique} is an
unstable clique that does not properly contain any other unstable
clique.  We can now define flexible cliques:

\begin{definition}[flexible clique]\label{def:flexible-clique}
We call a subset of neurons $\sigma \subset \{1,...,n\}$ a {\em flexible clique} of a network architecture on $n$ neurons, $(J,D)_G$, if for every $\varepsilon>0$ there exist $\varepsilon$-perturbations $A_s$ and $A_u$, consistent with $G$, such that $\sigma$ is a maximally stable clique of $(J+A_s,D)_G$ and a minimally unstable clique of $(J+A_u,D)_G$.  
\end{definition}

\noindent Flexible cliques are our model for flexible memory patterns.
All flexible cliques are marginal cliques, but the converse is not
true (see Section~\ref{sec:flex-marg}).  This is because the flexibility of a marginal clique
depends on the relationship of this clique to other cliques in the network.  In general, it is difficult to determine whether or 
not a marginal clique is flexible in a network with many marginal cliques.  We are interested in precisely this case,
as we look for properties of networks with the maximal number of flexible cliques.

\subsection{Statement of the Main Results}

Our main results all concern threshold-linear networks only.  Consequently, from now on we assume $\sN(n)$ and $\sN(G)$ correspond to sets of unconstrained and $G$-constrained threshold-linear networks, respectively.
First, we define what we mean by the `flexibility' and `rank' of a network:

\begin{definition}[network flexibility, rank, and completion]\label{def:net-flex}
We define the {\em flexibility} of a network as the number of flexible cliques, and denote it: $\mathrm{flex}(J, D)_G$.
We define the {\em rank} of a network $(J,D)_G$ to be the rank of the matrix $-D+J$.  We say that a $G$-constrained network on $n$ neurons, $(J,D)_G$, has a rank $k$ {\em completion} if there exists a network  $(\bar{J},\bar{D}) \in \sN(n)$ of rank $k$ such that $\bar{D} = D$ and $\bar{J}_{ij} = J_{ij}$ for all $i = j$ and all distinct pairs $(ij) \in G$. 
\end{definition}

\noindent We now further refine our main question:

\begin{main-question}[version 3]
  For a given constraint graph $G$, what threshold-linear networks $(J,D)_G \in \sN(G)$ attain maximum
  flexibility?
\end{main-question}

\noindent Note that the flexibility of a $G$-constrained network $(J,D)_G$ is bounded by the total number of non-empty cliques in the corresponding clique complex $X(G)$, and by the fact that single neurons can not be flexible cliques because $-D+J$ has strictly negative diagonal.  
Thus,
$$\mathrm{flex}(J,D)_G \leq |X(G)|-n-1,$$
where $n$ is the number of neurons.  In $\sN(n)$, the flexibility can be at most $2^n-n-1$.  Most networks, however, have no flexible cliques.

The rank of any network $(J,D)_G$ is at least 1. 
For threshold-linear networks, rank 1 networks are good candidates for attaining maximum flexibility because all but the $1 \times 1$ principal submatrices are marginally stable.   Indeed, we find that rank 1 networks attain the upper bound on flexibility in $\sN(n)$, and a similar statement is true about $G$-constrained networks:

\begin{theorem}\label{thm:rank1}
\thmV
\end{theorem}

\noindent The proof is given in Section~\ref{sec:thm2}.

Can any networks other than ones that are rank 1, or have rank 1 completions, attain maximal flexibility?  The following example demonstrates that a $G$-constrained network can be maximally flexible without having a rank 1 completion.

\begin{example}
Let $G = (V,E)$ be a simple graph with vertices $V = \{1,2,3,4\}$ and
edges $E = \{(12),(23),(34),(41)\}$.  $G$ is a cycle on $4$ vertices,
so the clique complex $X(G)$ has no cliques of size greater than $2$.  Consider the threshold-linear network $(J,D)_G \in \sN(G)$, where 
$$-D+J = \left(\begin{array}{cccc} -1 & 2 & 0 & 1 \\ 1/2 & -1 & 1 & 0\\ 0 & 1 & -1 & 1\\ 1 & 0 & 1 & -1 \end{array}\right).$$
Using  Lemma~\ref{lemma:marg-flex} (see Section~\ref{sec:flex-marg}), it is easy to see that all $\sigma \in X(G)$ such that $|\sigma| = 2$ are flexible cliques, and hence $(J,D)_G$ is maximally flexible in $\sN(G)$.  Despite this, $(J,D)_G$ does not have a rank 1 completion, since there is no rank 1 matrix that agrees with $-D+J$ on all of its nonzero entries (cf. Example~\ref{ex:top} in Section~\ref{ssec:rank1}).
\end{example}

\noindent It turns out, however, that examples of this kind can be eliminated by imposing a simple topological condition on the clique complex of the constraint graph $G$.  Note that a clique complex is an abstract simplicial complex, whose homology groups can be computed using simplicial homology.  The following is our main result:

\begin{theorem}\label{thm:main-result}
\thmIII
\end{theorem}

\noindent The vanishing of $H_1(X(G);\ZZ)$ may
at first appear to be a strong condition, but in fact it is generically satisfied
for large random networks that are not overly sparse.  For example, it was recently shown in 
\cite{Kahle2009} that if $G$ is an Erd\"os-R\'enyi graph with edge probability $p$ (i.e., a random graph on $n$ vertices with independent connection probability $p$ between any pair of vertices), then $p\geq n^{-\alpha}$ with $\alpha<1/3$
implies that  the probability of $H_1(X(G);\ZZ) = 0$ approaches $1$ as $n\to \infty$.  For $n = 10^4$ neurons, the first homology group of the clique complex is expected to vanish for connection probabilities as low as $p = .05$.

The proof of Theorem~\ref{thm:main-result} is given in Section~\ref{sec:thm34}.  
For the complete graph $K_n$, $H_1(X(K_n),\ZZ) = 0,$ thus the following result for
unconstrained networks is a corollary of Theorems~\ref{thm:rank1} and \ref{thm:main-result}:

\begin{theorem}\label{thm:unconstrained}
\thmII
\end{theorem}

\noindent We give in
Section~\ref{sec:thm34} a separate proof of this theorem
independent of homology/cohomology arguments.

\subsection{Discussion}
In this article we have laid out the foundations for a theory of flexible memory networks -- that is, for recurrent networks with memory patterns that can be both encoded (learned) and unencoded (forgotten) by arbitrarily small perturbations of the matrix of connection strengths between neurons.  Given a constraint graph $G$ of allowed connections, we have found, modulo a mild topological condition, that maximally flexible networks in $\sN(G)$ correspond precisely to networks $(J,D)_G$ that have a rank $1$ completion.  These results may provide valuable insights for understanding fast-learning, flexible networks in the brain.  

Our results are based on an analysis of the fixed point attractors of a standard firing rate model~\eqref{eq:threshlin}.  We emphasize, however, that we regard this only as a model of the {\em fast-timescale dynamics} of a recurrent network; in a more comprehensive model, additional elements such as stochastic fluctuations, changing external inputs, and adaptive variables on a slower timescale will all lead to frequent transitions between attractors of the fast-timescale equations~\eqref{eq:threshlin}.  This approach has proven particularly fruitful in modeling of hippocampal networks, where ``bump attractor'' models on a fast timescale form integral building blocks to more comprehensive models that have been successful in describing experimental observations from simultaneously recorded populations of neurons \cite{Samsonovich:1997, McNaughton:2006:NN,Romani:2010,JN2011}.

Thus far we have only considered the extreme case of maximally
flexible networks.  For these networks, arbitrarily small perturbations of the synaptic weights between neurons are sufficient to encode new memory patterns.  Nevertheless, larger perturbations will be necessary for these patterns to be robust in the presence of various plasticity mechanisms that are engaged during ongoing spontaneous activity.  For any given learning rule and/or constraint on the rate of synaptic changes, however, maximally flexible networks have the best chance to quickly encode (or unencode) new memories as stable fixed points of the dynamics.

It may be possible to extend these results to
other flexible networks.  Do all rank $k$ networks have the same
flexibility?  What is the general relationship between the flexibility
of a network and its rank?  
What is the appropriate generalization of the topological condition in
Theorem~\ref{thm:main-result} when considering higher rank
completions?  We leave these questions to future work.

The remainder of this paper is organized as follows.  In Section~\ref{sec:stable-set} we prove Theorem~\ref{thm:stable-clique}, making the connection between fixed points of the recurrent network dynamics and the stability of principal submatrices.  In Section~\ref{sec:matrices} we develop some matrix-theoretic results that are critical to proving our main theorems.  This section is self-contained, independent of the context of neural networks.  In Section~\ref{sec:max-flex} we prove our main results, Theorems~\ref{thm:rank1}, \ref{thm:main-result} and \ref{thm:unconstrained}.


\vspace{.1in}
\section{Stable sets correspond to stable principal submatrices} \label{sec:stable-set}
Recall that a threshold-linear network $(W,D)$ has dynamics described by the system,
\vspace{-.02in}
\begin{equation}\label{eq:thresh}
 \dot x = -Dx+\left[W x + b \right]_+,
\end{equation}
with $x\in \IR^n_{\geq 0}$ the firing rate vector, $D$ a diagonal matrix with strictly positive diagonal entries, and $-D+W$ an $n \times n$ matrix having strictly negative diagonal entries.
For such networks, one is able to obtain qualitative
characterizations (stable, marginal, unstable) of sets of neurons.  
In what follows, we consider fixed points of~\eqref{eq:thresh} for a fixed input vector $b \in \RR^n$.

Suppose there exists a fixed point of~\eqref{eq:thresh} with all neurons firing, i.e. $x^*>0$.  Since $Dx^*>0$, we can drop the threshold in a small neighborhood of this fixed point, where the system is described by the linear system $\dot x = (-D +W) x + b$.
If the matrix $-D + W$ is invertible, the linear system has
exactly one fixed point, $(D-W)^{-1}b$, although this fixed point may or may not be in the regime $\IR^n_{>0}$
where all neurons are firing. Either $x^*=(D-W)^{-1}b>0$, 
or we have a contradiction and there is no $x^*$.
As is well-known for linear systems, the fixed point $x^*$ is asymptotically stable if
and only if the matrix $-D+W$ is a stable matrix.

In addition to a possible fixed point with all neurons firing, the system~\eqref{eq:thresh} may also have fixed points
corresponding to proper subsets of neurons with non-zero
firing rate. Let $\sigma = \supp(x^*) \subset \{1,...,n\}$
be the subset of neurons that are firing at the fixed point $x^*$,
with the complement
$\bar{\sigma}$ representing the remaining (silent) neurons.
To describe these types of fixed points, we reorder the neurons and write
$$W = \mtx{W_{\bsigma\bsigma}}{W_{\bsigma\sigma}}{W_{\sigma\bsigma}}{W_{\sigma\sigma}}, \;\;
D = \mtx{D_{\bsigma}}{0}{0}{D_\sigma}, \;\;
x = \vec{x_\bsigma}{x_\sigma}, \mbox{~~and~~}\; b = \vec{b_\bsigma}{b_\sigma}.$$
The system~\eqref{eq:thresh} becomes,
\begin{eqnarray*}
\dot{x}_\bsigma &=& -D_\bsigma x_\bsigma + [W_{\bsigma \bsigma} x_\bsigma + W_{\bsigma \sigma} x_\sigma + b_\bsigma]_+,\\
\dot{x}_\sigma &=& -D_\sigma x_\sigma + [W_{\sigma \bsigma} x_\bsigma + W_{\sigma \sigma} x_\sigma + b_\sigma]_+, 
\end{eqnarray*}
and, since $x^*_\bsigma = 0$, the fixed point equations for $x^*$ simplify to:
\begin{eqnarray*}
0 &=& [W_{\bsigma \sigma} x^*_\sigma + b_\bsigma]_+,\\
D_\sigma x^*_\sigma &=& [W_{\sigma \sigma} x^*_\sigma + b_\sigma]_+.
\end{eqnarray*}
Since $D_\sigma x^*_\sigma > 0$, we can drop the threshold in the second equation and obtain
\begin{equation}\label{eq:fixedpt}
(D_\sigma - W_{\sigma\sigma})x^*_\sigma = b_\sigma.
\end{equation}
However, a solution $x^*$ of this equation only yields a valid fixed point if $x^*_\sigma>0$ and
$W_{\bsigma \sigma} x^*_\sigma + b_\bsigma \leq 0.$

To analyze the stability of a fixed point $x^*$ with $\supp(x^*) = \sigma$, we make the following change of coordinates. Let $\vec{y}{z} \od x-x^*$, with $y \in \RR^{|\bsigma|}_{\geq 0}$ and $z \in \RR^{|\sigma|}_{\geq 0}$.  Then $x^*$ is a stable fixed point of~\eqref{eq:thresh} if and only if the origin is a stable fixed point of:
 \begin{eqnarray*}
\dot{y} &=& -D_\bsigma y + [W_{\bsigma \bsigma} y + W_{\bsigma \sigma} z + (W_{\bsigma \sigma} x^*_\sigma + b_\bsigma)]_+,\\
\dot{z} &=& -D_\sigma(z + x^*_\sigma) + [W_{\sigma \bsigma} y + W_{\sigma \sigma} z + (W_{\sigma \sigma} x^*_\sigma + b_\sigma)]_+.
 \end{eqnarray*}
The existence of the fixed point $(y=0,z=0)$ implies that
$W_{\bsigma \sigma} x^*_\sigma + b_\bsigma \leq 0 $ and $W_{\sigma \sigma} x^*_\sigma + b_\sigma > 0.$
If we further assume that $W_{\bsigma \sigma} x^*_\sigma + b_\bsigma < 0 $, then there exists an open neighborhood of the origin for which the sign of each of the thresholded terms is determined by the constant terms (those that do not involve $y$ or $z$).  In this neighborhood, we can simplify the thresholds and, using~\eqref{eq:fixedpt}, the equations take the form,
\vspace{-.1in}
 \begin{eqnarray*}
\dot{y} &=& -D_\bsigma y\\
\dot{z} &=& -D_\sigma z + W_{\sigma \bsigma} y + W_{\sigma \sigma} z.
 \end{eqnarray*}
Because the system is exactly linear in a neighborhood of the fixed point, $x^*$ is asymptotically stable if and only if the matrix 
$$M = \mtx{-D_\bsigma}{0}{W_{\sigma \bsigma}}{-D_\sigma+W_{\sigma \sigma}}$$
is stable.  Similarly, $x^*$ is stable but not asymptotically stable if and only if $M$ is marginally stable, and $x^*$ is an unstable 
fixed point if and only if $M$ is unstable.  
Finally, note that the stability of $M$ is equivalent to the stability of $-D_\sigma+W_{\sigma \sigma} = (-D+W)_\sigma$.

We collect these observations into the following
characterization of fixed points in threshold-linear networks:

\begin{proposition}\label{prop:stable-set}
Consider the system~\eqref{eq:thresh},
for a threshold-linear network $(D,W)$ on $n$ neurons with fixed input $b$, and let
$\sigma \subset \{1,...,n\}$ be a subset of neurons. The following statements hold:
\begin{itemize}
\item[(i)] A point $x^*$ with $\supp(x^*) = \sigma$ is a fixed point if and only if $x^*_\sigma$ satisfies:
\begin{itemize}
\item[(a)] $(D_\sigma - W_{\sigma\sigma})x^*_\sigma = b_\sigma$, 
\item[(b)] $x^*_\sigma>0$, and
\item[(c)] $b_\bsigma \leq -W_{\bsigma \sigma} x^*_\sigma .$
\end{itemize}
In particular, if $\det(D_\sigma-W_{\sigma \sigma}) \neq 0$, then there exists at most one fixed point with support $\sigma$, and it is given by $x^*_\sigma = (D_\sigma - W_{\sigma\sigma})^{-1}b_\sigma.$
\item[(ii)] Suppose $x^*$ is a fixed point with $\supp(x^*) = \sigma$.  If $b_\bsigma < -W_{\bsigma\sigma}x^*_\sigma$, then $x^*$ is asymptotically stable if and only if the principal submatrix $(-D+W)_\sigma$ is stable.  Similarly, $x^*$ is stable but not asymptotically stable if and only if $(-D+W)_\sigma$ is marginally stable, and $x^*$ is an unstable fixed point if and only if $(-D+W)_\sigma$ is unstable.
\end{itemize}
\end{proposition}

\noindent Using this Proposition, we can now prove Theorem~\ref{thm:stable-clique}.

\begin{proof}[{\bf Proof of Theorem~\ref{thm:stable-clique}}]
We begin with the first statement.  ($\Rightarrow$) Let $\sigma$ be a stable set of $(W,D)$, and choose $b$ such that there exists an asymptotically stable fixed point $x^*$ of ~\eqref{eq:thresh}  with $\supp(x^*) = \sigma$.  By part (i) of Proposition~\ref{prop:stable-set} it is clear that we can choose $b$ such that $b_\bsigma < - W_{\bsigma \sigma} x^*_\sigma$.
It then follows from part (ii) that $(-D+W)_\sigma$ is stable.

$(\Leftarrow)$ Now suppose that $(-D+W)_\sigma$ is stable.  We construct $b$, the input vector
  in~\eqref{eq:thresh}, so that the corresponding fixed point with support $\sigma$ is asymptotically stable.
 Let $b_\sigma = (D-W)_\sigma 1_{\sigma}$, where
  $1_\sigma$ is the vector of all ones.  Letting $x^*$ be the firing rate
  vector with $\supp(x^*) = \sigma$ and $x^*_\sigma = 1_\sigma > 0$, we choose $b_\bsigma$
  to satisfy $b_\bsigma < -W_{\bsigma\sigma} x^*_\sigma$.   Note that $(D_\sigma - W_{\sigma\sigma})x_\sigma^* = b_\sigma$.  For this choice of $b$, it thus follows from part (i) of Proposition~\ref{prop:stable-set} that $x^*$ is a fixed point, and by part (ii) that $x^*$ is asymptotically stable.  Hence, $\sigma$ is a stable set of $(W,D)$.

Similar arguments using Proposition~\ref{prop:stable-set} can be used to show that $\sigma$ is a marginal or unstable set of $(W,D)$ if and only if $(-D+W)_\sigma$ is marginally stable or unstable, respectively.  
\end{proof}

As a result of Theorem~\ref{thm:stable-clique}, we see that in order to investigate
stable, unstable, or marginal sets of neurons in threshold-linear networks we need to understand the stability of
principal submatrices.


\vspace{.1in}
\section{Matrix-theoretic results}\label{sec:matrices}

In this section, we prove results 
concerning real matrices with strictly negative entries on the
diagonal.  These results will be critical for
Section~\ref{sec:max-flex}, where we prove our main
theorems regarding maximally flexible networks.  This section is
self-contained, however, and the results can be understood
independently of the context of neural networks.

Throughout this section, $A$ is an $n\times n$ matrix with real
coefficients and strictly negative entries on the diagonal.  
The matrix $\sE_{A}=(\epsilon_{ij})$ is the sign matrix associated to $A$; 
this is a matrix whose entries $\epsilon_{ij} \in \{\pm 1, 0\}$ are the signs
of the corresponding entries of $A$.
To the matrix $A$ we also associate the graph $G_A$, which we call the {\em
  connectivity graph} of $A$.  It is the simple graph that includes each edge $(ij)$ unless $A_{ij} = A_{ji} = 0$.
Let $X(G_A)$ be the clique complex associated to the graph $G_A$; we
call this the {\em clique complex associated to} the matrix $A$.  Note that a clique complex
is an abstract simplicial complex.  For any simplicial complex $X$ and abelian group $\mathcal{G}$, we denote the associated simplicial homology and cohomology groups as $H_i(X; \mathcal{G})$ and $H^i(X;\mathcal{G})$, respectively.  Finally, mirroring Definition~\ref{def:net-flex}, we call an $n \times n$ matrix $\bar{A}$ a {\em completion} of  $A$ if $\bar{A}_{ij} = A_{ij}$ for all $i = j$ and all distinct pairs $(ij) \in G_A$.

\subsection{Bipartite matrices}\label{ssec:bm}

Bipartite matrices play an important role in Section~\ref{ssec:rank1}.

\begin{definition}[bipartite matrix]\label{defn:bipartite} We say that a real-valued, $n \times n$ matrix $A$ is a {\em bipartite matrix}
if the index set $\{1,...,n\}$ can be partitioned into two disjoint sets, $\sigma$ and $\bar\sigma$, such that:
\begin{enumerate}
\item if $i \in \sigma$ and $j \in \bar\sigma$, both $A_{ij} \geq 0$ and $A_{ji} \geq 0$.
\item if $i,j \in \sigma$ or $i,j \in \bar\sigma$, both $A_{ij} \leq 0$ and $A_{ji} \leq 0$.
\end{enumerate}
\end{definition}

\noindent This definition is equivalent to the condition that there exists
a permutation of the indices $\{1,...,n\}$ such that the sign pattern
of $A$ takes on the block-form:
\begin{equation}\label{eq:block-form}
\mathrm{sgn}(A) = \left(\begin{array}{c|c} - & + \\ \hline + & - \end{array}\right), 
\end{equation}
where ``$+$'' indicates a submatrix with all nonnegative entries, and ``$-$'' a submatrix with all nonpositive entries.  From this observation it is easy to see that all rank $1$ matrices with negative diagonal are bipartite:

\begin{lemma}\label{lemma:signpattern}
Let $A$ be a real $n \times n$ matrix with $\rank A = 1$ and $A_{ii}<0$ for all $i=1,...,n$.  Then there exists a permutation of the indices such that the sign pattern of $A$ is of the form~\eqref{eq:block-form}.
In particular, $A$ is a bipartite matrix.
\end{lemma}

\noindent The following result gives a sufficient condition for bipartiteness of a matrix in terms of the associated clique complex.
\begin{lemma}[bipartite lemma] \label{lemma:bipartite1} Let $A$ be a
  real-valued $n \times n$ matrix with strictly negative diagonal, sign matrix $\E_A =
  (\epsilon_{ij})$, connectivity graph $G_A$, and clique complex
  $X(G_A)$. Suppose that
\begin{enumerate}
\item[(i)] $\epsilon_{ij}\epsilon_{ji} = 1$, whenever $(ij) \in X(G_A)$, 
\item[(ii)] $\epsilon_{ij}\epsilon_{jk}\epsilon_{ki} = -1$, whenever $(ijk) \in X(G_A)$, and
\item[(iii)]  $H^1(X(G_A);\ZZ_2) = 0.$
\end{enumerate}
Then $A$ is a bipartite matrix.
\end{lemma}  

\begin{proof}
It is convenient to think of $\ZZ_2 = \{\pm 1\}$, the multiplicative
group with two elements. Consider the co-chain complex, with $X=X(G_A)$,
\begin{equation}\label{chain-z2}
{\mathcal C}^{0}(X; \ZZ_2) \stackrel{\delta_0}\longrightarrow {\mathcal C }^1(X;\ZZ_2) \stackrel{\delta_1}{\longrightarrow}  {\mathcal C }^2(X;\ZZ_2)  \stackrel{\delta_2}{\longrightarrow} \cdots  \stackrel{\delta_{n-1}}{\longrightarrow} {\mathcal C }^n(X;\ZZ_2) \stackrel{\delta_n}{\longrightarrow} 0.
\end{equation}
The maps are the usual coboundary operators.  For example, $\delta_0 (\{v_i\}) = \{ e_{ij}\}$, where $e_{ij} = v_j v_i^{-1} = v_iv_j$.  Similarly, $\delta_1 (\{e_{ij}\}) = \{f_{ijk}\}$, where 
$f_{ijk} = e_{jk}e_{ik}^{-1}e_{ij} = e_{ij}e_{jk}e_{ki}.$

By (i) we have $\{\epsilon_{ij}\} \in {\mathcal C }^1(X;\ZZ_2) $,
while (ii) implies that $\{-\epsilon_{ij}\} \in \Ker \delta_1$
(we need the minus sign because kernel elements map to $+1$).  Using
(iii), we conclude that $\{-\epsilon_{ij}\} \in \Im \delta_0$.  Hence, there exists a vertex labeling $\{\nu_i\} \in {\mathcal C}^{0}(X; \ZZ_2)$ such that $-\epsilon_{ij} = \nu_i\nu_j$ whenever $(ij) \in X$.  Let $\sigma = \{i \:|\: \nu_i = +1\}$ and $\bar\sigma = \{i \:|\: \nu_i = -1\}$, with $\sigma \cup \bar\sigma = \{1,...,n\}$.  The sign of an edge, $\epsilon_{ij} = -\nu_i\nu_j$, can only be positive if $i \in \sigma$ and $j \in \bar\sigma$ or if $i \in \bar\sigma$ and $j \in \sigma$, and $\epsilon_{ij}$ can only be negative if $i,j \in \sigma$ or $i,j \in \bar\sigma$.  This proves that the matrix $A$ is bipartite.
\end{proof}

\begin{example}  To see why we need the cohomology condition in cases where there are zero entries, consider the matrix
\begin{small}
$$A = \left(\begin{array}{ccccc}-1 & 1 & 0 & 0 & 1 \\ 1 & -1 & 1 & 0 & 0\\ 0 & 1 & -1 & 1 & 0\\ 0 & 0 & 1 & -1 & 1\\ 1 & 0 & 0 & 1 & -1\end{array}\right).$$
\end{small}

\noindent
The graph $G_A$ is a cycle on 5 vertices, and A is not a bipartite matrix.
Nevertheless, $A$ satisfies conditions (i) and (ii).  
\end{example}

When the matrix $A$ has no zero entries, however, this kind of example cannot occur.  The clique complex of the complete graph, $X(K_n)$, is contractible, so $H^{1}(X(K_n);\ZZ_2) = 0$ and we obtain the following corollary:

\begin{corollary} \label{cor:bipartite}
Let $A$ be a real-valued, $n \times n$ matrix with strictly negative diagonal and sign matrix $\E_A = (\epsilon_{ij})$ with all entries $\epsilon_{ij}$ nonzero.  Suppose that
(i) $\epsilon_{ij}\epsilon_{ji} = 1$, for all $(ij)$, and
(ii) $\epsilon_{ij}\epsilon_{jk}\epsilon_{ki} = -1$, for all distinct triples $(ijk)$.
Then $A$ is a bipartite matrix.
\end{corollary}


\subsection{Vanishing principal minors and rank}\label{ssec:rank1}

Recall that, given an $n \times n$ matrix $A$ and a subset $\sigma \subset
  \{1,...,n\}$, the {\em principal submatrix} $A_\sigma$ is the square
  matrix obtained by restricting $A$ to the index set $\sigma$.
  The determinant of a principal submatrix is called a {\em principal minor}.

It is well-known that if the rank of $A$ is $k$, then all principal
minors larger than $k\times k$ must vanish.  The converse is not true.
For example, consider a strictly upper-triangular $n \times n$ matrix.
It can have rank up to $n-1$, yet each and every principal minor vanishes.
In the case of matrices with strictly negative diagonal entries, however, we do
have a kind of converse; this is the subject of Lemma~\ref{lemma:rank1} and 
Proposition~\ref{prop:sparse-rank1}, the main result in
  this section.  We begin with a simple technical lemma that
will be used throughout.

\begin{lemma}\label{lemma:zero-minors}
 Let $A$ be a real-valued $n \times n$ matrix with strictly negative diagonal
 entries and clique complex $X(G_A)$.  Suppose that
 all $2 \times 2$ and $3 \times 3$ principal minors corresponding to cliques in $X(G_A)$ vanish.
 Then the matrix $A$ has symmetric sign matrix $\sE_A$, and its entries satisfy
 \begin{eqnarray}
  A_{ii}A_{jj} &=& A_{ij}A_{ji}, \;\;\; \mathrm{for} \;\;\; (ij) \in X(G_A), \;\;\; \mathrm{and} \label{eq:2by2}\\
  A_{ii}A_{jj}A_{kk} &=& A_{ij}A_{jk}A_{ki}, \;\;\; \mathrm{for} \;\;\; (ijk) \in X(G_A).  \label{eq:3by3}
\end{eqnarray}
\end{lemma}

\begin{proof}
The first set of relations~\eqref{eq:2by2} is obvious, and ensures
that $\sE_A$ is symmetric and $A_{ij} \neq 0$ for all $(ij) \in X(G_A)$.  
This, together with the vanishing of $3 \times 3$ principal minors, yields:
$$2A_{ii}A_{jj}A_{kk} = A_{ij}A_{jk}A_{ki} + A_{kj}A_{ji}A_{ik}, \;\;\; \mathrm{for} \;\;\; (ijk) \in X(G_A).$$
Given a triple $(ijk) \in X(G_A)$, denote $x =
A_{ii}A_{jj}A_{kk}$ and $y = A_{ij}A_{jk}A_{ki},$ and note that $y
\neq 0$.  Using again~\eqref{eq:2by2}, we can write
$A_{kj}A_{ji}A_{ik} = x^2/y.$ The above $3 \times 3$ condition becomes
$2xy = y^2 + x^2,$ and we conclude that $x = y$, which yields~\eqref{eq:3by3}.
\end{proof}

We can now show that the vanishing of all $2\times 2$ and $3\times 3$
principal minors suffices to guarantee that $A$ has rank $1$.

\begin{lemma}[rank 1]\label{lemma:rank1}
Let $A$ be a real-valued, $n \times n$ matrix with strictly negative diagonal entries such that
 all $2 \times 2$ and $3 \times 3$ principal minors vanish.
Then $G_A$ is the complete graph, $A$ is a bipartite matrix, and $\rank A = 1$.
\end{lemma}

\begin{proof}
Since all $2\times 2$ principal minors vanish, it follows that 
for each $i \neq j$ we have
$ A_{ij} A_{ji} = A_{ii} A_{jj}$ $\neq 0$, and $G_A$ is therefore the complete graph.
By Lemma~\ref{lemma:zero-minors}, the relations
\eqref{eq:2by2} and \eqref{eq:3by3} are satisfied for all pairs and
triples of distinct indices.  Anchoring ourselves on
the first row and column, we find that any entry of the matrix can be
written as:
$$A_{ij} = \dfrac{A_{ii}A_{jj}A_{11}}{A_{j1}A_{1i}} = \dfrac{A_{ii}A_{1j}}{A_{1i}} = \dfrac{A_{i1}A_{1j}}{A_{11}}.$$
Let $u = (A_{11},A_{21},...,A_{n1})^T$ be the first column vector of $A$ and $v = (A_{11},A_{12},...,A_{1n})$ the first row vector.  
Then $A = (A_{11})^{-1} u v$, which is manifestly rank 1.  It follows from Lemma~\ref{lemma:signpattern} that $A$ is bipartite.
\end{proof}

Can we generalize this result for
matrices with zeroes -- i.e., for matrices $A$ such that $G_A$ is not the complete graph?  Here we are looking for conditions that ensure the matrix $A$ has a rank 1 {\em completion}, 
where the entries with zeroes are treated as ``unknown'' entries that can be completed to any value.
In this case, we can require only that all $2 \times 2$ and $3 \times 3$ principal minors in the clique complex $X(G_A)$ vanish.  The following example shows that such a requirement is insufficient to guarantee the existence of a rank 1 completion.

\begin{example}\label{ex:top}  Consider the matrix
$$A = \left(\begin{array}{cccc} -1 & a & 0 & 1/d\\ 1/a & -1 & b & 0\\ 0 & 1/b & -1 & c\\ d & 0 & 1/c & -1 \end{array}\right).$$
This matrix has $G_A = (V,E)$, where $V = \{1,2,3,4\}$ and $E =
\{(12),(23),(34),(41)\}$.  $G_A$ is a cycle on $4$ vertices, and
the clique complex $X(G_A) = G_A$ since there are no $2$-dimensional faces.  Note that all $2 \times 2$ principal minors 
corresponding to $2$-cliques in $X(G_A)$ vanish, and there are no $3\times 3$ ones to check.  Does this matrix have a rank $1$ completion?  Generically, the answer is ``No."  In fact, it is easy to see that a rank 1 completion exists if and only if $abcd = 1$.  
\end{example}

The intuition we gain from this example is that there is a topological
obstruction to a matrix having a rank $1$ completion.  It is the
presence of a closed but hollow cycle in $X(G_A)$ that prevents $A$
from having a rank $1$ completion.  In fact, if we added or removed an
edge from the graph $G_A$ in Example~\ref{ex:top}, we would have a
rank 1 completion without any further condition other than the
vanishing of $2\times 2$ and $3 \times 3$ principal minors $\det A_\sigma$ for $\sigma \in X(G_A)$.  
The following Proposition gives topological
conditions that guarantee the existence of a rank $1$ completion.
Note that a condition ensuring that $A$
is bipartite is needed to show that $A$ (as opposed to only $|A|$)
has a rank $1$ completion.


\begin{proposition}\label{prop:sparse-rank1}
Let $A$ be a real-valued $n \times n$ matrix with strictly negative diagonal and clique complex
$X(G_A).$ Let $|A|$ denote the matrix of absolute values of $A$.  Suppose that
$\det A_\sigma = 0$ for all $\sigma \in X(G_A)$ such that $|\sigma| = 2$ or $3$.
Then,
\begin{enumerate}
\item[(a)] $H^1(X(G_A);\ZZ_2) = 0 \implies A$ is a bipartite matrix.
\item[(b)] $H^{1}(X(G_A);\RR) = 0 \implies$ $|A|$ has a rank 1 completion.  
\item[(c)] $H^{1}(X(G_A);\RR) = H^1(X(G_A);\ZZ_2) = 0 \implies$ $A$ has a rank $1$ completion.
\end{enumerate}
\end{proposition}

\begin{proof}
Let $X= X(G_A)$.
  First, observe that we satisfy the conditions of Lemma~\ref{lemma:zero-minors}, and so we have relations~\eqref{eq:2by2} and~\eqref{eq:3by3}.  

\noindent (a):  
Let $\sE_{A}= (\epsilon_{ij}),$ with $\epsilon_{ij} \in \{\pm 1,0\}$, be the sign matrix of $A$.  Relations~\eqref{eq:2by2} and~\eqref{eq:3by3} imply
\begin{eqnarray*}
\epsilon_{ij}\epsilon_{ji} &=& 1, \;\;\; \mathrm{for} \;\;\; (ij) \in X, \;\;\; \mathrm{and}\\
\epsilon_{ij}\epsilon_{jk}\epsilon_{ki} &=& -1, \;\;\; \mathrm{for} \;\;\; (ijk) \in X.  
\end{eqnarray*}
Since we also have $H^1(X;\ZZ_2) = 0$, it follows from Lemma~\ref{lemma:bipartite1} that $A$ is a bipartite matrix.

\noindent(b):
For every $A_{ij}$ that is non-zero, introduce the following (real) variables:
\[
  L_{ij}: = \ln \left(\frac{|A_{ij}|}{\sqrt{A_{ii}A_{jj}}}\right).
\]
In these variables, the relations~\eqref{eq:2by2} and~\eqref{eq:3by3} are equivalent to antisymmetry and cocycle conditions on the $L_{ij}$:
\begin{eqnarray} 
\label{antisym} L_{ij}+L_{ji} &=& 0,\;\;\; \mathrm{for} \;\;\; (ij) \in X, \;\;\; \mathrm{and} \\
\label{co-cycle}  L_{ij}+L_{jk}+L_{ki} &=& 0, \;\;\; \mathrm{for} \;\;\; (ijk) \in X.
\end{eqnarray}
Now consider the co-chain complex
\begin{equation}\label{chain}
{\mathcal C}^{0}(X; \RR) \stackrel{\delta_0}\longrightarrow {\mathcal C }^1(X;\RR) \stackrel{\delta_1}{\longrightarrow}  {\mathcal C }^2(X;\RR)  \stackrel{\delta_2}{\longrightarrow} \cdots  \stackrel{\delta_{n-1}}{\longrightarrow} {\mathcal C }^n(X;\RR) \stackrel{\delta_n}{\longrightarrow} 0.
\end{equation}
where $C^{k}(X;\RR)$ is the group of $k$-cochains with coefficients in $\RR$. $C^{0}(X;\RR)$ corresponds to vertex-labelings, $C^{1}(X;\RR)$ is the set of edge-labelings, etc.  As usual, the coboundary operators are $\delta_k (\{f_{i_0,...,i_k}\}) = \{g_{i_0,...,i_{k+1}}\}$, where
$$g_{i_0,...,i_{k+1}} = \sum_{j=0}^{k+1} (-1)^{j} f_{i_0,...,\widehat{i}_j,...,i_{k+1}},$$
and $\delta_{k+1} \circ \delta_k = 0$.  By assumption, $H^{1}(X;\RR) = 0,$ so $\Im \delta_0 = \Ker \delta_1$.

Let $L=(L_{ij}), \mbox{\;for\;} (ij) \in X.$ Observe that \eqref{antisym}
implies that $L \in   {\mathcal C }^1(X;\RR),$  while the cocycle
condition  \eqref{co-cycle} implies  that $L\in \Ker
\delta_1.$ It follows that $L \in \Im \delta_0,$
so there exists an $a\in  {\RR }^n \cong  {\mathcal C}^{0}(X;\RR)$ such that $L_{ij}=a_j-a_i.$ This implies  that, for each $A_{ij} \neq 0$,
\[
  |A_{ij}|=\sqrt{A_{ii}A_{jj}}e^{L_{ij}}=\sqrt{|A_{ii}|}e^{-a_i}\sqrt{|A_{jj}|}e^{a_j}.
\]
Let $u,v \in \RR^n$ with $u_i = \sqrt{|A_{ii}|}e^{-a_i}$ and $v_j = \sqrt{|A_{jj}|}e^{a_j}$.  Since $|A|$ and $uv^T$ agree on on all nonzero entries of $|A|$,  the matrix $\overline{|A|} = uv^T$ is a rank 1 completion of $|A|$.

\noindent(c):
Recall from the proof of part (a) that $\epsilon_{ij}$ is the sign of $A_{ij}$, so that $A_{ij} =
\epsilon_{ij} |A_{ij}|$ for each entry of $A$.   Following the proof of Lemma~\ref{lemma:bipartite1}, there exists a vertex labeling $\{\nu_i\} \in \CC^0(X;\ZZ_2)$, with $\nu_i \in \{\pm 1\}$, such that $\epsilon_{ij} = -\nu_i\nu_j$ whenever $(ij) \in X$.  Choose $u,v \in \RR^n$ as in the proof of part (b), so that $|A_{ij}| = u_iv_j$ whenever $(ij) \in X$.  Now consider $\tilde{u},\tilde{v} \in \RR^n$ where $\tilde{u}_i = -\nu_iu_i$ and $\tilde{v}_j = \nu_jv_j$.  It follows that $A_{ij} = \tilde{u}_i\tilde{v}_j$ whenever $A_{ij} \neq 0$.  The matrix $\bar{A} = \tilde{u}\tilde{v}^T$ is thus a rank 1 completion of $A$.
\end{proof}

\begin{remark}
  Note that Lemma~\ref{lemma:rank1} follows easily from Proposition~\ref{prop:sparse-rank1}, 
  since the clique complex of the complete graph $X(K_n)$ is contractible, so the conditions
  $H^1(X(K_n);\ZZ_2) = 0$ and $H^1(X(K_n);\RR)=0$ are trivially
  satisfied.
In Theorem~\ref{thm:main-result}, for simplicity we use instead the
somewhat stronger condition $H_1(X(G);\ZZ) = 0$.  If $H_1(X(G);\ZZ) = 0$, then $H^1(X(G);\ZZ_2) =
H^1(X(G);\RR) = 0$; this follows from the following well-known
observation.
\end{remark}

\begin{lemma}\label{lem:vanish-Hl1} 
 Let $X$ be a simplicial complex. Assume that $H_{1}(X;\IZ) =0$. Then
 $H^1(X,\mathcal{G}) =0$, for every abelian group $\mathcal{G}$.
\end{lemma}
\begin{proof}
 This is a consequence of the Universal Coefficients Theorem \cite{Hatcher}, which
 for an abelian group $\mathcal{G}$ yields the short exact
 sequence
 \begin{equation*}
  0\longrightarrow \text{Ext}(H_{q-1}(X,\Z),\mathcal{G})
  \longrightarrow H^{q}(X,\mathcal{G})\longrightarrow
  \text{Hom} (H_q(X,\Z),\mathcal{G})\longrightarrow 0
 \end{equation*}
 for all $q\geq 1$. Note that for $H$ a free abelian group, 
 $\text{Ext} (H,\mathcal{G}) =0$. Since $H_0(X,\Z)$ is always free, the above for $q=1$ 
 yields $H^{1}(X,\mathcal{G}) \cong  \text{Hom} (H_1(X,\Z),\mathcal{G})=0$.
\end{proof}


\subsection{Stable and marginally stable matrices}\label{sec:stable}

Recall that all flexible cliques are marginal cliques, and by Theorem~\ref{thm:stable-clique} the marginal cliques correspond to marginally stable principal submatrices.  Therefore, to make the connection to flexible cliques in Section~\ref{sec:max-flex}, we need to consider what happens to a matrix when its principal submatrices are marginally stable, which is not quite the same as having vanishing determinant.

Recall that a matrix is marginally stable if no eigenvalue has strictly positive real part and at least one eigenvalue is purely imaginary.  
In the case of symmetric matrices, marginal stability implies the
existence of a zero eigenvalue, and hence vanishing determinant.
This is not in general true for non-symmetric matrices.  However, $2\times 2$ and $3\times 3$ marginally stable matrices with negative diagonal entries do have the following characterization:

\begin{lemma}\label{lem:23}
 (i) Let $A$ be a $2\times 2$ real matrix with strictly negative diagonal entries. 
 Assume that $A$ is marginally stable. Then $\det (A) =0$,
 and the sign matrix $\sE_A$ is a symmetric matrix with all entries nonzero. \\
 (ii) Let $A$ be a $3\times 3$ real matrix with strictly
 negative diagonal entries. Assume that $A$ is marginally
 stable. Then either $\det(A) =0$, or $\det(A) \neq 0$ and $A$ has a $2\times 2$ stable principal submatrix.
\end{lemma}

\begin{proof}
(i) The matrix $A$ must have a purely
  imaginary eigenvalue. Since $\tr(A)<0$, this eigenvalue must be $0$, and thus
  $\det(A) =0$. As a consequence, the off-diagonal entries of the
  matrix $A$ must have the same sign and are both nonzero. 

 \noindent
 (ii) Let $\lambda_1$ be a
 purely imaginary eigenvalue of $A$. We have two possibilities: either 
 $\lambda_1 =0$ and thus $\det(A)=0$, or $\lambda_1 \neq 0$.  In the second case, the conjugate 
  $\bar{\lambda}_1 = -\lambda_1$ is also an eigenvalue,
 and since $\tr(A) <0$, the third eigenvalue $\lambda_3$ must be negative and so
 $\det(A) \neq 0$. 
 Consider now the characteristic polynomial of $A$,
 $P_{A}(X)=-X^{3} + \Tr(A) X^{2} - M_{2}(A) X + \det (A),$
  where $M_{2}(A)$ denotes the sum of the principal $2 \times 2$ minors of $A$.
 Using the usual expression for the coefficients of $P_A(X)$ in terms of eigenvalues of $A$, we find that
$M_2(A) = \lambda_1 (-\lambda_1) + \lambda_1\lambda_3 +(- \lambda_1)\lambda_3 = |\lambda_1|^2 >0.$
 There thus exists a $2 \times 2$ principal submatrix of $A$ with
 positive determinant and negative trace.  This submatrix is stable.
\end{proof}

We also have relationships between the stability of a matrix and its principal submatrices.  In the case of symmetric matrices, it follows from Cauchy's interlacing theorem that all principal submatrices of a stable matrix are stable.  

\begin{theorem}[Cauchy's interlacing theorem] Let $A$ be a symmetric $n \times n$ matrix, and let $B$ be an $m \times m$ principal submatrix of $A$, where $m \leq n$.  If the eigenvalues of $A$ are $\alpha_1 \leq ...\alpha_j... \leq \alpha_n$, and those of $B$ are $\beta_1 \leq ... \beta_j ... \leq \beta_m$, then for all $j$ we have $\alpha_j \leq \beta_j \leq \alpha_{n-m+j}.$
\end{theorem}

\begin{corollary}\label{cor:simplicialprop}
Any principal submatrix of a stable symmetric matrix is stable.  Any symmetric matrix containing an unstable principal submatrix is unstable.
\end{corollary}

\noindent Even in the case of non-symmetric matrices, there are still some constraints of this type.  For example,

\begin{lemma}\label{lemma:2x2}
Let $A$ be an $n \times n$ matrix with strictly negative diagonal
entries and $n \geq 2$.  If $A$ is stable, then there exists a $2 \times 2$ principal submatrix of $A$ that is also stable.
\end{lemma}
\begin{proof}
 We use the formula for the characteristic polynomial in terms of sums of principal minors:
 $$
   P_A(x) = (-1)^n x^n + (-1)^{n-1}M_1(A) x^{n-1}+
   (-1)^{n-2}M_2(A) x^{n-2}+ \ldots + M_n(A),
 $$
 where $M_k(A)$ is the sum of the $k \times k$ principal minors of
 $A$.  (Note that $M_1(A) = \Tr(A)$ and $M_n(A) = \det(A)$.)
 The characteristic polynomial also has the
 well-known formula with coefficients that are symmetric polynomials in the eigenvalues; assuming $A$ is stable,
 this yields
 $ M_2(A) = \sum_{i<j} \lambda_i\lambda_j > 0.$
 This implies that at least one $2 \times 2$ principal minor is
 positive.  Since the corresponding $2 \times 2$ principal submatrix
 has negative trace, it must be stable.
\end{proof}

\noindent In order to prove our main results in Section~\ref{sec:max-flex}, we will also use the following well-known consequences of Cauchy's interlacing theorem. Here $A_k$ refers to the principal submatrix obtained by taking the upper left $k \times k$ entries of $A$.

\begin{lemma}[Stable symmetric matrices]\label{lemma:signcondition}  Let $A$ be a real symmetric $n \times n$ matrix.  Then $A$ is stable iff $(-1)^k\det(A_k)>0$ for all $1 \leq k \leq n$.
\end{lemma}

\begin{corollary}\label{cor:signcondition}
Let $A$ be a real symmetric $n \times n$ matrix.  Then $A$ is stable iff $(-1)^{|\sigma|}\det(A_\sigma)>0$ for every principal submatrix $A_\sigma$.
\end{corollary}


\vspace{.1in}
\section{Maximally flexible networks}\label{sec:max-flex}

In this section we use the matrix results from Section~\ref{sec:matrices} in order to prove our main results, Theorems~\ref{thm:rank1}, \ref{thm:main-result} and \ref{thm:unconstrained}, characterizing maximally flexible networks.

\subsection{Flexible vs. marginal cliques} \label{sec:flex-marg}

Recall that all flexible cliques are marginal cliques, because they can be made both stable and unstable via arbitrarily small perturbations of the network's connection strengths.  The converse is not true.
The following lemma gives simple, but incomplete, conditions for determining whether or not a marginal clique is flexible
in threshold-linear networks.

\begin{lemma}\label{lemma:marg-flex}
Let $\sigma$ be a marginal clique of a threshold-linear network $(J,D)_G$.  
\begin{enumerate}
\item If there exists $\tau\in X(G)$ such that either (i)
$\tau \subsetneq \sigma$ and $\tau$ unstable, or (ii) $\tau \supsetneq \sigma$ and $\tau$ stable, then $\sigma$ is not a flexible clique.
\item If, on the other hand, (i) for all $\tau \subsetneq \sigma$, $\tau$ is a stable clique, and
(ii) for all $\tau \in X(G)$ such that $\tau \supsetneq \sigma$, $\tau$ is an unstable clique,
then $\sigma$ is a flexible clique.
\end{enumerate}
\end{lemma}

\noindent The proof follows from observing that any marginal clique can be perturbed to become stable or unstable by adding a multiple of the identity matrix to the corresponding principal submatrix, and one can always find a small enough perturbation so that the stability of all stable and unstable principal submatrices in the original matrix is preserved.
It is thus straightforward to check the flexibility of marginal cliques if certain patterns of stable/unstable cliques are also present.  This is illustrated in the following example.

\begin{example} Consider the following matrices $-D+J$ for (unconstrained) threshold-linear networks $(D,J) \in \sN(3)$:  
\begin{small}
$$M_1 = \left(\begin{array}{ccc}-1 & 0 & -2 \\ -2 & -1 & 0 \\ 0 & -2 & -1 \end{array}\right),\;
M_2 = \left(\begin{array}{ccc}-1 & -1 & 1 \\ -1 & -1 & 0 \\ 0 & 1 & -1 \end{array}\right), \;\mbox{and}\;\;
M_3 = \left(\begin{array}{ccc}-1 & 2 & 1 \\ 1 & -1 & 0 \\ 0 & -1 & -1 \end{array}\right).
$$
\end{small}
\begin{itemize}
\item[$M_1$:] $\{1,2,3\}$ is a flexible clique since it is marginal and all contained cliques are stable.
\item[$M_2$:] $\{1,2\}$ is a marginal clique but it is not flexible, since $\{1,2,3\}$ is stable.
\item[$M_3$:] $\{1,2,3\}$ is a marginal clique but it is not flexible, since $\{1,2\}$ is unstable.
\end{itemize}
\end{example}

Note that Lemma~\ref{lemma:marg-flex} says nothing about the situation
where the cliques contained by or containing a given marginal clique
are themselves also marginal.  It is much more difficult to
check for flexible cliques in a network with many marginal cliques.
We investigate precisely this case, as we look for
properties of networks with the maximal number of flexible cliques.

\subsection{Proof of Theorem~\ref{thm:rank1}}\label{sec:thm2}

We begin with an example of a network in which all the cliques with at
least two neurons are flexible. Such a network is maximally
flexible, and provides a reference point in proving
that all rank $1$ networks are maximally flexible in $\sN(n)$.
The proof relies on the following determinant formulas:

\begin{lemma}
Let $W_n(\epsilon,\alpha)$, for $n \geq 2$, be the symmetric $n \times n$ matrix with entries
\begin{eqnarray}\label{eq:Wn}
W_n(\epsilon,\alpha)_{ij} = \left\{\begin{array}{ccc} -1, &\mathrm{if}& i = j,\\
-1+\alpha\epsilon &\mathrm{if}& \{i,j\} = \{1,2\}, \\
-1 + \epsilon &\mathrm{if}& \{i,j\} \neq \{1,2\}.\end{array}\right.
\end{eqnarray}
Then,
\begin{eqnarray}\label{eq:det}
\hspace{.3in} \det W_n(\epsilon,\alpha) = (-1)^n \alpha\epsilon^{n-1}\left(2n-2 - (2n-4)\epsilon - (n-2-(n-3)\epsilon)\alpha\right).
\end{eqnarray}
In particular,
\begin{eqnarray}\label{eq:det2}
\hspace{.3in} \det W_n(\epsilon,1) = (-1)^n \epsilon^{n-1}\left(n - (n-1)\epsilon\right).
\end{eqnarray}
\end{lemma}

\begin{proof}
This is a straightforward determinant computation.
\end{proof}

We now show that the matrix  with all entries $-1$ corresponds to a network on $n$ neurons 
that has the maximal number $2^n-n-1$ of flexible cliques, and is thus maximally flexible.

\begin{proposition}\label{prop:saturation}
 Let $(J,D)  \in \sN(n)$ be the network with the matrix $-D +J = -\one$, where $-\one$ is the $n \times n$ matrix having
 all entries $-1$.
 Then any subset $\sigma \subset \{1, \ldots, n\}$ with at least two
 neurons is a flexible clique. In particular,
 $\mathrm{flex}(J,D) = 2^n-n-1$.
\end{proposition}

\begin{proof}
 Let $\sigma$ be any subset with $|\sigma| = k\geq 2$ neurons.
To show that $\sigma$ is flexible,
 it suffices to show that there exists an $\e_0 >0$ so that for every
 $0<\e<\e_0$, there exist $\varepsilon$-perturbations $A_s$ and $A_u$
 of $(J,D)$ under which $\sigma$ becomes a maximally stable clique and a
 minimally unstable clique, respectively.  We show this via explicit
 construction of $A_s$ and $A_u$.

 Let $A_s$ be the symmetric matrix with $0$ entries on the diagonal,
 entries $(A_s)_{ij} = \varepsilon$ for distinct $i,j \in \sigma$, and
 $(A_s)_{ij} = -\varepsilon$ if either $i \notin \sigma$ or $j \notin
 \sigma$.  Clearly, $A_s$ is an $\varepsilon$-perturbation.  We need
 to show that $\sigma$ is a maximally stable clique for $(J +
 A_s,D)$; i.e., $\sigma$ is a stable
 clique of the perturbed network, and any clique $\tau$ which properly contains $\sigma$ is
 unstable. By Theorem~\ref{thm:stable-clique}, it is enough to show
 that the corresponding principal submatrices of $-\one +A_s$ are stable
 and unstable, respectively.

Recall~\eqref{eq:Wn}, and note that the principal submatrix $(-\one+A_s)_{\sigma} = W_k(\varepsilon,1)$, where
$k = |\sigma|$.  Using~\eqref{eq:det2}, we obtain
\begin{equation}\label{eq:s-sigma}
 \det (-\one +A_s)_{\sigma} = (-1)^k \e^{k-1}(k - (k-1) \e).
\end{equation}
Note that the same expression holds for any $\sigma' \subset \sigma$, with $k = |\sigma'|$.
By Corollary~\ref{cor:signcondition}, it follows that $(-\one+A_s)_{\sigma}$ is stable for all $0<\e\leq 1$.
To show that $\sigma$ is maximally stable, observe that any clique $\tau$ properly containing $\sigma$ 
must also contain an order $2$ clique whose corresponding principal submatrix is
$$\left(\begin{array}{cc} -1 & -1-\varepsilon \\ -1 - \varepsilon & -1 \end{array}\right),$$
which is unstable for all $\varepsilon>0$.
Since the matrix $(-\one + A_s)_\tau$ is symmetric, it follows from
Corollary~\ref{cor:simplicialprop} that $(-\one
+A_s)_{\tau}$, for any $\tau \supsetneq \sigma$,  is unstable for all $0< \e \leq 1$.

To generate a perturbation $A_u$ for which the clique $\sigma$ is minimally unstable,
we proceed as follows. Let $0<\e \leq 1$, and choose two
neurons $i_1, i_2 \in \sigma$ such that $i_1 = \min (\sigma)$ and $i_2 = \min(\sigma - \{i_1\})$.  
Let $A_u$ be the symmetric matrix with entries $(A_u)_{ij} = \varepsilon$
for distinct $i,j \in \sigma$ unless $\{i,j\}= \{i_1,i_2\}$.  We let
the entries $(A_u)_{i_1 i_2} = (A_u)_{i_2 i_1} = \alpha \e$, with
$\alpha$ to be determined later.  All other entries of $A_u$ are set to $0$.
To show that $\sigma$ is minimally unstable, we
need to choose $\alpha$ so that $(-\one + A_u)_{\sigma}$ is unstable while all its
proper principal submatrices are stable.  Since $-\one+A_u$ is symmetric, Corollary~\ref{cor:signcondition} tells us that this 
is accomplished if the determinant of $(-\one + A_u)_{\sigma}$ has the `wrong' sign $(-1)^{k+1}$, where $k = |\sigma|$, and
all $j \times j$ principal minors of  $(-\one + A_u)_{\sigma}$, with $j < k$, have the `right' sign $(-1)^j$.

Observing that $(-\one+A_u)_{\sigma} = W_k(\varepsilon,\alpha)$, we have from~\eqref{eq:det} that 
\begin{equation}\label{eq:sigma}
\det(-\one+A_u)_\sigma
= (-1)^k\alpha\varepsilon^{k-1}(2k-2-(2k-4)\varepsilon-(k-2-(k-3)\varepsilon)\alpha).
\end{equation}
There are two types of proper principal submatrices. 
The first are those that correspond to the cliques $\tau \subsetneq \sigma$ that
contain both $i_1$ and $i_2$, with $j = |\tau|$, and are equal to the matrices $W_j(\varepsilon,\alpha)$.  From~\eqref{eq:det}, these
have determinants
\begin{equation}\label{eq:tau}
\det(-\one +A_u)_{\tau}
= (-1)^j\alpha\varepsilon^{j-1}(2j-2-(2j-4)\varepsilon-(j-2-(j-3)\varepsilon)\alpha).
\end{equation}
The second type of principal submatrices correspond to cliques 
$\nu \subsetneq \sigma$ that do not contain both $i_1$ and $i_2$.  Letting $j = |\nu|$, these
are equal to the matrices $W_j(\varepsilon,1)$, and by~\eqref{eq:det2} have determinants
\begin{equation}\label{eq:nu}
\det(-\one +A_u)_{\nu}=
(-1)^{j}\varepsilon^{j-1}(j - (j-1) \varepsilon).
\end{equation}
Using Corollary~\ref{cor:signcondition}, we see that
the cliques of type $\nu$ are all stable for $0<\varepsilon \leq 1$.
It remains to choose $\alpha$ so that~\eqref{eq:sigma} has sign $(-1)^{k+1}$ and~\eqref{eq:tau} has
sign $(-1)^j$ for all $j = 2,...,k-1$.

For $k>3$, we choose $\alpha$ so that 
$$
 \frac{2k-2 - (2k-4)\e}{k-2 - (k-3)\e} < \alpha < \min\left\{\frac{2j-2 - (2j-4)\e}{j-2 - (j-3)\e}\right\}_{j=2,...,k-1}.
$$ 
This is always possible, since for $0< \varepsilon \leq 1$ the
sequence on the right is decreasing; the minimum is attained for $j =
k-1$ and is greater than the term on the left, corresponding to $j =
k$.  Since for $k >3$ we also have $\alpha \e < 2\frac{k-2}{k-3}\e$ 
the matrix $A_u$ is a $4\e$-perturbation. When $k=3$, one
can choose $0<\e<\frac{1}{2}$ and $\alpha =4$, while in the case $k=2$
one needs simply to choose $\alpha <0$ so that $\sigma$ is a minimally
unstable clique.
\end{proof}

\noindent We now show that all the symmetric rank $1$ networks are maximally flexible:

\begin{proposition}\label{prop:sym-rank1}
Let $(J,D) \in \sN(n)$ be a
symmetric rank $1$ network.  Then $(J,D)$ is maximally flexible in
$\sN(n)$ and has flexibility $2^n-n-1$.
\end{proposition}

\begin{proof}
  Recall that by definition $(J,D) \in \sN(n)$ is a symmetric rank $1$
  network if the matrix $-D +J$ is a symmetric rank $1$ matrix. Since
  the matrix $-D + J$ has negative entries on the diagonal, there exists 
  a vector $x \in \IR^n$ so that $-D+J = -xx^T$.
  Let $\diag(x)$ be the $n\times n$ diagonal matrix associated to the
  vector $x\in \IR^n$. Then
  \begin{equation*}
    - D + J = - x x^T = \diag(x) \, (-\one) \diag(x),
  \end{equation*}
  where $-\one$ is the $n\times n$ matrix with all entries $-1$
  that we encountered in Proposition~\ref{prop:saturation}.
  
Since multiplication of a matrix on the left and right
  by the same diagonal matrix does not alter the sign of any principal minor, we have for any perturbation $A$
  $$\det(-xx^T+\diag(x)A \diag(x))_\sigma = \det(-\one+A)_\sigma,$$
  for any $\sigma \subset \{1,...,n\}$.
  Moreover, if $A$ is a symmetric perturbation, then so is $\diag(x)A \diag(x)$, and the
  stability of any principal submatrix of $-\one+A$ or $-xx^T+\diag(x)A \diag(x)$ is determined entirely by the signs of the principal minors (Corollary~\ref{cor:signcondition}). 
  We can thus obtain stable and unstable perturbations
  $\tilde{A}_s$ and $\tilde{A}_u$ of $-xx^T$
  for any subset $\sigma$ consisting of $|\sigma|\geq 2$
  neurons by modifying the perturbations $A_s$ and $A_u$ in
  Proposition~\ref{prop:saturation} accordingly: $\tilde{A}_s = \diag(x)
  A_s \diag(x)$ and $\tilde{A}_{u} = \diag(x) A_u \diag(x)$.
  We conclude that the network $(J,D)$ is maximally flexible in $\sN(n)$, with flexibility $2^n-n-1$.
\end{proof}

Before proving Theorem~\ref{thm:rank1}, which extends the above
results to $G$-constrained nonsymmetric networks in $\sN(G)$, we need to define the notion of `pruning' of a network.  We say that a graph $\tilde{G}$ is a {\em pruning} of the graph $G$ if the two graphs have the same vertices and the edges of $\tilde{G}$ form a subset of the edges of $G$. We say that a network $(\tilde{J},\tilde{D})_{\tilde{G}}$ is a {\em pruning} of $(J,D)_G$ if $\tilde{G}$ is a pruning of $G$, $\tilde{D} = D$, and $\tilde{J}_{ij} = J_{ij}$ for all edges $(ij) \in \tilde{G}$.  The following lemma shows that flexible cliques are `inherited' by pruning.

\begin{lemma}\label{lem:pruning} Let  $(\tilde{J},D)_{\tilde{G}}$ be a pruning of $(J,D)_G$.  Consider a clique $\sigma \in X(\tilde{G}) \subset X(G)$.  If $\sigma$ is a flexible clique of $(J,D)_G$, then $\sigma$ is also a flexible clique of $(\tilde{J},D)_{\tilde{G}}$.
\end{lemma}
\begin{proof}
  This follows from the Definition~\ref{def:flexible-clique} of flexible
  cliques. If $\sigma \in X(G)$ is a flexible clique, then there exist
  perturbations $A_s$ and $A_u$ consistent with $G$ so that $\sigma$
  is maximally stable for $(J+A_s,D)_G$ and minimally unstable for
  $(J+A_u, D)_G$.  If we also have $\sigma \in X(\tilde{G})$, define the perturbations $\tilde{A}_s$ and
  $\tilde{A}_u$, consistent with $\tilde G$, by setting all the entries
  in $A_s$ and $A_u$ corresponding to the pruned edges
  to $0$.  Since $X(\tilde{G}) \subset X(G)$, it follows that the perturbations
  $\tilde{A}_s$ and $\tilde{A}_u$ realize $\sigma$ as a flexible
  clique of $(\tilde{J},D)_{\tilde{G}}$.
\end{proof}

\noindent We now have all the ingredients necessary for proving Theorem~\ref{thm:rank1}. 


\begin{proof}[{\bf Proof of Theorem~\ref{thm:rank1}}] 
 We prove first that all rank $1$ networks $(J,D) \in \sN(n)$ are
 maximally flexible. Since the matrix $-D +J$ has rank $1$ and
 negative entries on the diagonal, there exists two vectors $x,y \in
 \IR^n$, with $x_i y_i>0$ for all $i = 1,...,n$, so that  $-D+J = -xy^T$.  
 Using these two vectors, we construct the diagonal matrix  $d =
 \diag\left(\sqrt{\dfrac{y_i}{x_i}}\right)$. Let $P$ be the matrix
 obtained from $-D +J$ by conjugation with the matrix $d$, i.e. $P = d
 (-D +J) d^{-1}$. It has entries
 $$P_{ij} = \sqrt{\dfrac{y_i}{x_i}} (-x_iy_j)\sqrt{\dfrac{x_j}{y_j}} = -\sqrt{x_iy_i}\sqrt{x_jy_j},$$
 and is therefore a rank $1$ symmetric matrix. By
 Proposition~\ref{prop:sym-rank1} it follows that the network $(d J
 d^{-1}, d D d^{-1})$ is a maximally flexible network in
 $\sN(n)$.
 Since $P$ and $-D +J$ are similar matrices, related via conjugation by a
 diagonal matrix, it follows that all corresponding principal
 submatrices $P_\sigma$ and $(-D+J)_\sigma$ are also similar. Hence, a perturbation $A$
 of the network $(J,D)$ has exactly the same stable and unstable
 cliques as a perturbation $d A d^{-1}$ of the network $(d J d^{-1},
 dD d^{-1})$. Since $(d J d^{-1},
 dD d^{-1})$ is maximally flexible, it follows that $(J,D)$ is also a maximally flexible network in $\sN(n)$.

Now let $(J,D)_G \in \sN(G)$ be a $G$-constrained network with a
  rank $1$ completion. We can think of the graph $G$ as a pruning of
  the complete graph $K_n$ on $n$ vertices. Let $(\bar{J},D) \in
  \sN(n)$ be a rank $1$ completion of the network $(J,D)_G$. By the previous arguments, 
  the network $(\bar{J},D)$ is maximally
  flexible in $\sN(n)$, and has flexibility $2^n-n-1$. In
  particular, any clique $\sigma \in X(G)$ with $|\sigma| \geq 2$ is a
  flexible clique of the network $(\bar{J},D)$.  By Lemma~\ref{lem:pruning}, $\sigma$ is also
  a flexible clique of $(J,D)_G$. Since $X(G)$ has $|X(G)|-n-1$ cliques
  with more than two neurons, it follows that the flexibility of
  $(J,D)$ is $|X(G)|-n-1$, which is maximal. 
\end{proof}

\subsection{Proof of Theorems~\ref{thm:main-result} and \ref{thm:unconstrained}}\label{sec:thm34}

First we prove our main result, Theorem~\ref{thm:main-result}.

\begin{proof}[{\bf Proof of Theorem~\ref{thm:main-result}}]
 Let $(J,D)_G$ be a maximally flexible threshold-linear
 network in $\sN(G)$.  This means that all the cliques $\sigma \in
 X(G)$ with at least two neurons are flexible.
%
 Since all flexible cliques are
 marginal cliques, Theorem~\ref{thm:stable-clique} gives that the
 corresponding principal submatrices of $(-D+J)_\sigma$ are all marginally stable.
 In particular, all $2\times 2$ and $3\times 3$ principal
 submatrices are marginally stable, and thus by Lemma~\ref{lem:23} it follows that 
 \begin{equation*}
  \det (-D+J)_{\sigma} =0, \quad \text{for all $\sigma \in X(G)$ with
    $|\sigma| = 2$ or $3$.}
 \end{equation*}
 Applying Lemma~\ref{lemma:zero-minors} to $-D+J$, it follows that
 for all $(ij) \in G$, the entry $J_{ij} \neq 0$. Thus the network
 $(J,D)_G$ has no silent connections.  By Lemma~\ref{lem:vanish-Hl1},
 the homology condition $H_1(X(G);\IZ) =0$ implies that
 $H^1(X(G);\IZ_2) = H^1(X(G);\IR) = 0$, and then by
 Proposition~\ref{prop:sparse-rank1} it follows that the matrix $-D+J$ has
 a rank $1$ completion. At the level of networks, this translates to
 $(J,D)_G$ having a rank $1$ completion.
\end{proof}

Theorem~\ref{thm:unconstrained} states that for the set $\sN(n)$ of unconstrained threshold-linear networks, the maximally
flexible networks are exactly the rank $1$ networks.  The proof is a direct application of Theorems~\ref{thm:main-result} and \ref{thm:rank1}.

\begin{proof}[{\bf Proof of Theorem~\ref{thm:unconstrained}}]
  ($\Rightarrow$) This direction is a direct
consequence of
  Theorem~\ref{thm:main-result}. 
Let $(J,D) \in \sN(n)$ be a maximally
  flexible network. Its graph is the complete graph $K_n$,
  and thus the corresponding clique complex $X(K_n)$ is contractible and
  satisfies $H_1(X(K_n),\ZZ) = 0$.  By Theorem~\ref{thm:main-result} it follows that $(J,D)$ is
  a rank $1$ network.
 ($\Leftarrow$) This follows from first part of Theorem~\ref{thm:rank1}.
\end{proof}

We also give a second proof of Theorem~\ref{thm:unconstrained}, without appealing to the
homological arguments used in the proof of Theorem~\ref{thm:main-result}.

\begin{proof}[{\bf Proof of Theorem~\ref{thm:unconstrained} without  homology/cohomology}]
 ($\Rightarrow$)
Suppose $(J,D)$ $ \in \sN(n)$ is a maximally flexible network.  This means $(J,D)$ must have flexibility
$2^n-n-1$. In particular, all $2\times 2$ and
$3\times 3$ principal submatrices of $-D+J$ must be marginally stable, and so by Lemma~\ref{lem:23} all
$2\times 2$ and $3\times 3$ principal minors must vanish. 
This, together with the fact that the diagonal
entries are strictly negative, implies that $-D+J$ satisfies the hypotheses of Lemma~\ref{lemma:rank1}, whose proof
does not rely on cohomology arguments, and is thus rank $1$.
($\Leftarrow$) This follows from first part of Theorem~\ref{thm:rank1}, which does not use homology or cohomology arguments.
\end{proof}

\begin{small}
\begin{section}{Acknowledgments}
\noindent C.C. was supported by NSF DMS-0920845 and a Courant Instructorship.  
A.D. was partially supported by NSF DMS-0505767.
V.I. was supported by NSF DMS-0967377 and the Swartz Foundation.
\end{section}

\bibliographystyle{amsplain}

\bibliography{refs}
\end{small}

\end{document}